\documentclass[12pt,journal,onecolumn]{IEEEtran}

\usepackage{amssymb,amsmath,amsfonts,bm,epsfig,graphicx,theorem,latexsym}
\usepackage{rotating,setspace,latexsym,epsf,color,subfigure}
\usepackage{cite,authblk}
\newcommand{\es}[1]{E_{#1}}	
\newcommand{\er}[1]{\bar{E}_{#1}}	
\newcommand{\ps}[1]{P_{#1}} 
\newcommand{\pr}[1]{\bar{P}_{#1}} 
\newcommand{\rs}[1]{r_{#1}}	
\newcommand{\rr}[1]{\bar{r}_{#1}}	
\newcommand{\ds}[1]{\delta_{#1}} 
\newcommand{\ops}[1]{P_{#1}^*}		
\newcommand{\opr}[1]{\bar{P}_{#1}^*} 	
\newcommand{\ods}[1]{\delta_{#1}^*}	
\newcommand{\wls}[1]{\nu_{#1}}	
\newcommand{\wlnew}[1]{\tilde{\nu}_{#1}}	
\newcommand{\prnew}[1]{\tilde{P}_{#1}}	
\newcommand{\oprnew}[1]{\tilde{P}_{#1}^*}	
\newcommand{\esource}{\textbf{E}}	
\newcommand{\erelay}{\bar{\textbf{E}}}	
\newcommand{\psource}{\textbf{P}}	
\newcommand{\prelay}{\bar{\textbf{P}}}	
\newcommand{\htap}[2]{T_{#1 #2}}	
\newcommand{\vtap}[1]{Q_{#1}}	
\newcommand{\deltav}{\boldsymbol{\delta}} 
\newcommand{\mycap}{\mathcal{C} (\esource,\erelay)}	
\newcommand{\hp}{\theta} 
\newcommand{\rsource}{\textbf{r}^*}	
\newcommand{\rrelay}{\bar{\textbf{r}}^*}	
\newcommand{\ors}[1]{r_{#1}^*}	
\newcommand{\orr}[1]{\bar{r}_{#1}^*}	

\newtheorem{lemma}{Lemma}

\newenvironment{Proof}[1]{\medskip\par\noindent{\bf Proof:\,}\,#1}{{\mbox{\,$\blacksquare$}\medskip\par}}

\makeatletter

\begin{document}
\IEEEoverridecommandlockouts

\title{Energy Cooperation in Energy Harvesting Communications\thanks{Berk Gurakan, Omur Ozel and Sennur Ulukus are with the Department of Electrical and Computer Engineering, University of Maryland. Jing Yang is with the Department of Electrical Engineering, University of Arkansas. This work was supported by NSF Grants CNS 09-64632, CCF 09-64645, CCF 10-18185 and CNS 11-47811, and presented in part at the IEEE ISIT, Cambridge, MA, July 2012, IEEE Asilomar Conference, Pacific Grove, CA, November 2012, and IEEE ICC, Budapest, Hungary, June 2013.}}

\author{Berk Gurakan, Omur Ozel, Jing Yang and Sennur Ulukus}

\maketitle

\vspace*{-0.7in}

\begin{abstract}
In energy harvesting communications, users transmit messages using energy harvested from nature during the course of communication. With an optimum transmit policy, the performance of the system depends only on the energy arrival profiles. In this paper, we introduce the concept of {\it energy cooperation}, where a user wirelessly transmits a portion of its energy to another energy harvesting user. This enables shaping and optimization of the energy arrivals at the energy-receiving node, and improves the overall system performance, despite the loss incurred in energy transfer. We consider several basic multi-user network structures with energy harvesting and wireless energy transfer capabilities: relay channel, two-way channel and multiple access channel. We determine energy management policies that maximize the system throughput within a given duration using a Lagrangian formulation and the resulting KKT optimality conditions. We develop a {\it two-dimensional directional water-filling algorithm} which optimally controls the flow of harvested energy in two dimensions: in time (from past to future) and among users (from energy-transferring to energy-receiving) and show that a generalized version of this algorithm achieves the boundary of the capacity region of the two-way channel.
\end{abstract}

\section{Introduction}

In energy harvesting communications, users transmit messages using energy harvested from nature \cite{Yates09TWC, tassiulas10TWC, sharma10TWC}. In such systems, transmission policies of the users need to be carefully designed according to the energy arrival profiles. Recent work addresses this energy management problem for various energy harvesting communication settings \cite{tcom-submit, kaya_subm, ozel11, wless-submit, uysal_paper, finite, jing12jcn, kaya_jcn, Zhang_Relay, gunduz_camsap, orhanCISS, letaief, gunduz_jcn, kaya_subm2, orhan_itw, jiexu12}. When the energy management policies are optimized as in \cite{tcom-submit, kaya_subm, ozel11, wless-submit, uysal_paper, finite, jing12jcn, kaya_jcn, Zhang_Relay, gunduz_camsap, orhanCISS, letaief, gunduz_jcn, kaya_subm2, orhan_itw, jiexu12}, the resulting performance of the system depends only on the energy arrival profiles. In this paper, we introduce the notion of \textit{energy cooperation} in energy harvesting communications where users can share a portion of their harvested energy with the other users by means of wireless energy transfer \cite{Gurakan12isit, Gurakan12asilomar, Gurakan13icc}. This energy cooperation enables us to control and optimize the energy arrivals at users to the extent possible. In the classical setting of cooperation \cite{erkip03}, users help each other in the transmission of their data by exploiting the broadcast nature of wireless communications and the resulting overheard information. In contrast to the usual notion of cooperation, which is at the \textit{signal level}, energy cooperation we introduce here is at the \textit{battery energy level}. In a multi-user setting, energy may be abundant in one user in which case the loss incurred by transferring it to another user may be less than the gain it yields for the other user. It is this cooperation that we wish to explore in this paper for several basic multi-user scenarios, where energy can be transferred from one user to another through a separate wireless energy transfer unit.

Wireless energy transfer has been recently proposed as a promising technique for a wide variety of wireless networking applications \cite{Kaibin12sub, shi_infocom_11, doost10, Ferguson11, poon12commag, poon12transopt}. In future wireless networks, nodes are envisioned to be capable of harvesting energy from the environment and transferring energy to other nodes, rendering the network energy self-sufficient and self-sustaining with a significantly prolonged lifetime. Wireless energy transfer is a relatively new concept for wireless communications; however, it has been considered in other contexts earlier: Wireless powering of engineering systems by microwave power transfer technology has been used in many applications \cite{brown84aes, microwavespace, sahai11opt} for a long time, such as space missions \cite{microwavespace} and optical communications \cite{sahai11opt}. While microwave power transfer is viewed as the key technology for large-scale cellular networks \cite{Kaibin12sub}, recent advances in wireless energy transfer technology supports feasibility of wireless network design in smaller scales. In \cite{soljacic07Science, Soljacic08}, wireless energy transfer with strong inductive coupling has been demonstrated with relatively high efficiency over relatively long distances with small device sizes. Another related line of research in medical implanting applications has been presented in \cite{poon12transopt, poon12commag, Ferguson11} where wireless nodes are powered by wireless energy transfer, which also use the wirelessly transferred energy for communications. RFID technology is another prominent example along this direction, where nodes harvest received energy and use the harvested energy (via reflection) for communication \cite{Glover_RFID}. Relying on the possibility of efficient wireless energy transfer, in this paper, we investigate the optimum communication schemes in multi-user systems with nodes that have energy harvesting and energy transfer capabilities.

In communication systems with wireless energy transfer, energy and information flow simultaneously. Motivated by this nature of such systems, the trade-off between energy and information transmission has been addressed in several recent works \cite{grover, varshneyisit08, zhang11, varshneyisit12, Simeone12ITW, osvaldoCLT, Schober}. Among these works, the one that is most pertinent to our work is \cite{osvaldoCLT}, where multi-user communication systems with simultaneous energy and information transmission are studied. Our problem formulation captures a different trade-off than those studied in \cite{grover, varshneyisit08, zhang11, varshneyisit12, Simeone12ITW, osvaldoCLT, Schober} since in our model wireless energy transfer is maintained by a separate wireless energy transfer unit, and the harvested energy source is independent of the received signal energy.

In this paper, we study the offline optimal energy management problem for several basic multi-user network structures with energy harvesting transmitters and one-way wireless energy transfer. Offline throughput maximization problem has been recently investigated for various settings with energy harvesting transmitters in \cite{tcom-submit, kaya_subm, ozel11, wless-submit, uysal_paper, finite, jing12jcn, kaya_jcn, Zhang_Relay, gunduz_camsap, orhanCISS, letaief, gunduz_jcn, kaya_subm2, orhan_itw, jiexu12}. In \cite{tcom-submit}, transmission completion time minimization problem for an energy harvesting transmitter with an unlimited sized battery is solved, and this solution is extended to the case of a transmitter with a finite sized battery in \cite{kaya_subm} by showing its equivalence to a throughput maximization problem. References \cite{ozel11, wless-submit, uysal_paper, finite, jing12jcn, kaya_jcn} extend the throughput maximization problem and its solution to fading, broadcast, multiple access and interference channels. In \cite{Zhang_Relay, gunduz_camsap, orhanCISS, letaief}, the end-to-end throughput maximization problem is solved for two-hop cooperative relay networks for various settings. Extensions of the throughput maximization problem for nodes with battery imperfections are considered in \cite{gunduz_jcn, kaya_subm2}, and processing costs are incorporated in \cite{orhan_itw, jiexu12}.

As extensively emphasized in \cite{tcom-submit, kaya_subm, ozel11, wless-submit, uysal_paper, finite, jing12jcn, kaya_jcn, Zhang_Relay, gunduz_camsap, orhanCISS, letaief, gunduz_jcn, kaya_subm2, orhan_itw, jiexu12}, in energy harvesting transmitters, energy arrivals in time impose energy causality constraints on the transmission policies of the users. In the optimal policy, due to the concavity of the throughput in powers, energy needs to be allocated as constant as possible over time subject to energy causality constraints. In the presence of wireless energy transfer, energy causality constraints take a new form: energy can flow in time from the past to the future for each user, and from one user to the other at each time. This requires a careful joint management of energy flow in two separate dimensions, and different management policies are required depending on how users share the common wireless medium and interact over it. In this context, we analyze several basic multi-user energy harvesting network structures with wireless energy transfer. To capture the main trade-offs and insights that arise due to wireless energy transfer, we focus our attention on simple two- and three-user communication systems.

First, we examine additive Gaussian two-hop relay channel with one-way energy transfer from the source node to the relay node where the objective is to maximize the end-to-end throughput. Next, we consider the Gaussian two-way channel with one-way energy transfer, and the two-user Gaussian multiple access channel with one-way energy transfer. For these two channel models, we determine the two-dimensional simultaneously achievable throughput regions. For all three cases, we use a Lagrangian approach and determine the optimum transmit powers and energy transfer policies via the KKT optimality conditions. In particular, we develop a \textit{two-dimensional directional water-filling algorithm} which optimally controls the energy flow in time and among users. As observed in \cite{ozel11}, energy harvesting setting gives rise to a \textit{directional} water-filling algorithm, where energy can flow only from the past to the future due to the energy causality constraints. In addition, with wireless energy transfer, at any give time, energy can flow from one user to the other depending on the direction of wireless energy transfer. Therefore, the directionality of energy flow in two separate dimensions requires careful management of energy over time and users. Solutions obtained in each setting yield new insights on energy cooperation at the battery energy level in the presence of wireless energy transfer.

\section{Two-Hop Relay Channel with One-Way Energy Transfer}
\label{model}

In this section, we consider a two-hop relay channel consisting of a source node, a relay node and a destination node as shown in Fig.~\ref{sysmod1}. The two queues at the source and the relay nodes are the data and energy queues. The energies that arrive at the source and the relay nodes are saved in the corresponding energy queues. The data queue of the source always carries some data packets to be delivered to the destination. The data packets sent from the source node cause a depletion of energy from the source energy queue and an increase in the relay data queue. These data packets are then served out of the relay data queue with a cost of energy depletion from the relay energy queue. The relay operates in a full-duplex mode, i.e., the data and energy queues of the relay are updated simultaneously in every slot. We assume that the data and energy buffer sizes are unlimited. In addition, energy expenditure is only due to data transmissions; any other energy costs, e.g., processing, circuitry, are not considered in this paper. There is a separate wireless energy transfer unit at the source node, and therefore, the source node may wish to share a portion of its energy with the relay node so that the relay can forward more data.

\begin{figure}[t]
\begin{center}
\includegraphics[width=0.55\linewidth]{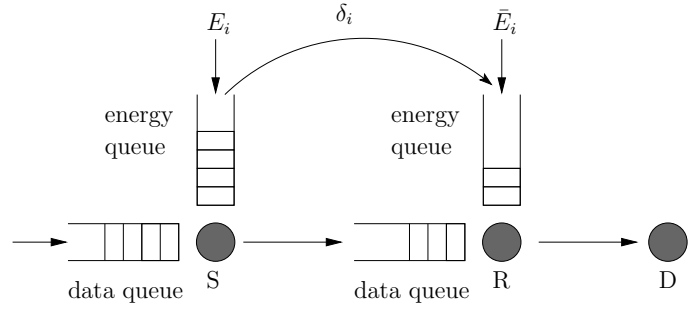}
\end{center}
\caption{Two-hop relay channel with energy harvesting source and relay nodes, and one-way energy transfer from the source node to the relay node.}
\label{sysmod1}
\vspace*{-0.1in}
\end{figure}

The channels from the source to the relay and from the relay to the destination are additive white Gaussian noise (AWGN) channels. The received signals $y_r$ and $y_d$ at the relay and the destination, respectively, are given by $y_r=\sqrt{h_s}x_s + n_s$ and $y_d=\sqrt{h_r}x_r + n_r$, where $h_s$ and $h_r$ are the channel coefficients for the source-to-relay and relay-to-destination channels, respectively. $n_s$ and $n_r$ are Gaussian noises each with zero-mean and unit-variance. We assume that $h_s = h_r = 1$ without loss of generality as otherwise the energy arrivals can be properly scaled.

Time is slotted and there are a total of $T$ equal length slots. Without loss of generality, we assume that the slots are of unit length. At times $t=1,\dots,T$, the source harvests energy with amounts $\es{1},\es{2},\dots,\es{T}$ and the relay harvests energy with amounts $\er{1},\er{2},\dots,\er{T}$. Energy transfer efficiency is $\alpha$, where $0 \leq \alpha \leq 1$. This means that when the source transfers $\ds{i}$ amount of energy to the relay through the wireless energy transfer unit in slot $i$, $\alpha \ds{i}$ amount of energy enters the energy queue of the relay in the next slot. Similarly, when the source uses power $\ps{i}$ for data transmission, the data queue of the relay is increased by $\frac{1}{2}\log\left(1+\ps{i}\right)$ bits in the next slot. The source and relay slots are indexed by one slot delay, so that, the slot subscripts are aligned at the source and the relay; see Fig.~\ref{slotmodel}. Power policy of the source is the sequences $\ps{i}$ and $\ds{i}$, and the power policy of the relay is the sequence $\pr{i}$.

As the energy that has not arrived yet cannot be used for data transmission or energy transfer, the power policies of the source and the relay are constrained by the causality of energy in time. These constraints yield the following feasible set:
\begin{align}
\mathcal{F} = \Big\{ (\deltav,\psource,\prelay): \ \ \sum_{i=1}^k \ps{i} \leq  \sum_{i=1}^k (\es{i} - \delta_i), \ \ \sum_{i=1}^k \pr{i} \leq  \sum_{i=1}^k (\er{i} + \alpha \delta_i), \ \ \sum_{i=1}^k \delta_i \leq \sum_{i=1}^k \es{i}, \ \ \forall k \label{feasset} \Big\}
\end{align}
where vectors $\psource$, $\prelay$ and $\deltav$ denote sequences $\ps{i}, \pr{i}$ and $\ds{i}$, respectively. $\mathcal{F}$ is the feasible set due to energy causality in harvested and transferred energies and is valid for the two-way and multiple access system models as well. For the two-hop relay channel model, we have an additional constraint: The relay transmits data that arrives from the source. Therefore, the power policies of the source and the relay need to satisfy the following data causality constraints at the relay:
\begin{align}
\sum_{i=1}^k \frac{1}{2} \log{(1+ \pr{i})} \leq \sum_{i=1}^k \frac{1}{2} \log{(1+ \ps{i})}, \quad k=1,\ldots, T
\end{align}
We formulate the end-to-end throughput maximization problem in the next section.

\begin{figure}[t]
\begin{center}
\includegraphics[width=0.45\linewidth]{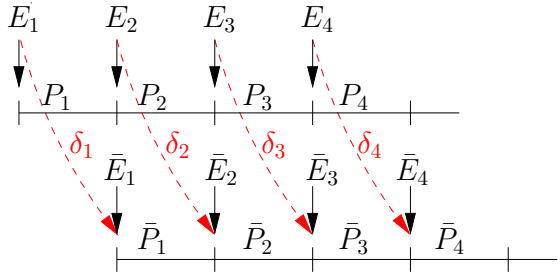}
\end{center}
\caption{Slotted system model: The queues of the relay are updated with one slot delay with respect to the queues of the source so that the slot indices are aligned.}
\label{slotmodel}
\vspace*{-0.1in}
\end{figure}

\section{End-to-end Throughput Maximization for the Relay Channel}

The optimal offline end-to-end throughput maximization problem with wireless energy transfer subject to energy causality at both nodes and data causality at the relay node is:
\begin{eqnarray}
\nonumber \max_{\pr{i}, \, \ps{i}, \, \ds{i}} & & \sum_{i=1}^T \frac{1}{2} \log{(1+\pr{i})} \\
\nonumber \mbox{s.t.} & & \sum_{i=1}^k \frac{1}{2} \log{(1+ \pr{i})} \leq \sum_{i=1}^k \frac{1}{2} \log{(1+ \ps{i})},
\quad k = 1,\dots,T \\ & & (\deltav,\psource,\prelay) \in \mathcal{F} \label{opt_probrelay1}
\end{eqnarray}
It can be shown that (\ref{opt_probrelay1}) is equivalent to a convex optimization problem (see \cite{Gurakan12isit}), by a change of variables from $\pr{i}, \ps{i}, \ds{i}$ to $\rr{i}=\frac{1}{2}\log\left(1+\pr{i}\right), \rs{i}=\frac{1}{2}\log\left(1+\ps{i}\right), \ds{i}$.  Thus, (\ref{opt_probrelay1}) can be solved using standard techniques \cite{boyd}. The Lagrangian function for the problem in (\ref{opt_probrelay1}) is:
\begin{align}
\nonumber \mathcal{L} =&- \sum_{i=1}^T \log{(1+\pr{i})} + \sum_{k=1}^T \lambda_k \left(\sum_{i=1}^k \log{(1+\pr{i})} - \sum_{i=1}^k \log{(1+\ps{i})} \right)	\\
&+ \sum_{k=1}^T \mu_k \left(\sum_{i=1}^k \ps{i} - (\es{i} - \delta_i) \right) + \sum_{k=1}^T \eta_k \left(\sum_{i=1}^k \pr{i} - (\er{i} + \alpha \delta_i) \right) - \sum_{k=1}^T \rho_k \delta_k \label{lagrangian_relaytwc}
\end{align}
Lagrange multiplier $\rho_k$ is due to the constraint that $\ds{k} \geq 0$. Note that the same constraints apply for $\ps{i}$ and $\pr{i}$; however, $\ps{i}$ and $\pr{i}$ are always non-zero in the optimal policy, therefore we exclude them. Similarly, we exclude the constraints $\sum_{i=1}^k \delta_i \leq \sum_{i=1}^k \es{i}$ in the Lagrangian function as these constraints can never be satisfied with equality in the optimal policy. This is because, if the source transfers all of its arrived energy at some slot, then the relay will have no data to send at that slot, and allocating a small portion of the arrived energy for source transmission strictly increases the throughput in this case. The KKT conditions for this problem are:
\begin{align}
\frac{-1 + \sum_{k=i}^T \lambda_k}{1 + \pr{i}} + \sum_{k=i}^T \eta_k &= 0, \quad \forall i\label{KKT_relay}\\
\frac{-\sum_{k=i}^T \lambda_k}{1 + \ps{i}} + \sum_{k=i}^T \mu_k &= 0,  \quad \forall i\label{KKT_source} \\
\sum_{k=i}^T \mu_k - \alpha \sum_{k=i}^T \eta_k - \rho_i &= 0,  \quad \forall i\label{KKT_energy}
\end{align}
with the additional complementary slackness conditions as:
\begin{align}
\lambda_k  \left(\sum_{i=1}^k \log{(1+\pr{i})}  - \sum_{i=1}^k \log{(1+\ps{i})} \right) &	= 0, \quad \forall k \\
\mu_k  \left(\sum_{i=1}^k \ps{i} - (\es{i} - \ds{i}) \right) & = 0, \quad \forall k \\
\eta_k  \left(\sum_{i=1}^k \pr{i} - (\er{i} + \alpha \ds{i}) \right) & = 0, \quad \forall k \\
\rho_k \ds{k} & = 0, \quad \forall k \label{comp_slack}
\end{align}
From (\ref{KKT_relay}), (\ref{KKT_source}) and (\ref{KKT_energy}) we get:
\begin{align}
\pr{i} =&   \frac{1 - \sum_{k=i}^T \lambda_k}{\sum_{k=i}^T \eta_k} -1, \quad \forall i\label{relay_soln} \\
\ps{i} =&  \frac{\sum_{k=i}^T \lambda_k}{\sum_{k=i}^T \mu_k} -1, \quad \forall i\label{source_soln} \\
\rho_i =& \sum_{k=i}^T \mu_k - \alpha \sum_{k=i}^T \eta_k, \quad \forall i\label{transfer_soln}
\end{align}
Next, we obtain necessary optimality conditions for (\ref{opt_probrelay1}).

\subsection{Necessary Optimality Conditions}

The first necessary optimality condition for (\ref{opt_probrelay1}) is that the source has to send as many bits as the relay can send and the relay has to finish up all the data in its data buffer. In other words, in the optimal policy, no data should be left in the data queue of the relay at the end.
\begin{lemma}
\label{opt_cond1}
The optimal power sequences $\ops{i}$, $\opr{i}$ must satisfy the constraint $\sum_{i=1}^T \frac{1}{2} \log(1+\opr{i}) = \sum_{i=1}^T \frac{1}{2} \log(1+\ops{i})$.
\end{lemma}
\begin{proof}
Suppose the stated constraint is satisfied with strict inequality. Then, we can increase some
$\ds{i}$, increase $\pr{i}$ and decrease $\ps{i}$ without violating the energy constraints and improve the overall throughput which contradicts the optimality of $\opr{i}$, $\ops{i}$, $\ods{i}$.
\end{proof}

We note that if the relay energy profile is sufficient to forward all the bits in the optimal source data stream with respect to the source energy profile, that is, if the separable policy in \cite{Zhang_Relay,gunduz_camsap} yields a policy that satisfies the necessary condition in Lemma \ref{opt_cond1}, then it is the optimal solution for (\ref{opt_probrelay1}) and no energy transfer is needed.

The second observation about the optimal policy is that the source has to exhaust the energies that have been harvested throughout the communication session either for data transmission or in the form of wireless energy transfer.
\begin{lemma}
\label{opt_cond2}
The optimal power profiles $\ops{i}$, $\opr{i}$ and energy transfers $\ods{i}$ must satisfy  $\sum_{i=1}^T \ops{i} = \sum_{i=1}^T (\es{i} - \ods{i})$.
\end{lemma}
\begin{proof}
Suppose this constraint is satisfied with strict inequality. Then, we can increase some $\ds{i}$ and then increase $\ps{i}$ and $\pr{i}$ to achieve a larger throughput and satisfy the constraints of (\ref{opt_probrelay1}). This contradicts the optimality of $\ops{i},\opr{i}, \ods{i}$.\end{proof}

Next, we observe that if there is a non-zero energy transfer from the source to the relay, then the relay has to exhaust all of its energy in the optimal policy.
\begin{lemma}
\label{opt_cond3}
For the optimal power sequences $\ops{i}$, $\opr{i}$ and energy transfer sequence $\ods{i}$, if $\ods{i} \neq 0$ for some $i$, then $\sum_{i=1}^T \opr{i} = \sum_{i=1}^T (\er{i} + \alpha \ods{i})$.
\end{lemma}
\begin{proof}
Suppose this constraint is satisfied with strict inequality. Using a similar argument as in Lemma \ref{opt_cond2}, we can decrease $\ds{i}$ and increase $\pr{i}$ to achieve a larger throughput and satisfy the constraints of problem (\ref{opt_probrelay1}). This contradicts the optimality of $\ops{i},\opr{i}, \ods{i}$.
\end{proof}

Finally, we note that, in the optimal policy, the total energy expenditure at the relay must be higher than the total energy expenditure at the source.
\begin{lemma}
\label{opt_cond4}
The optimal power sequences $\ops{i}$ and $\opr{i}$ must satisfy $\sum_{i=1}^T \ops{i} \leq \sum_{i=1}^T \opr{i}$, and with equality if and only if $\ops{i}=\opr{i}$ for all $i$.
\end{lemma}

\begin{proof}
We will give a proof based on majorization theory and Schur convexity \cite{majorbook}. We denote the optimal source and relay rate allocation vectors as $\rsource = [\ors{1}, \dots ,\ors{T}]$ and $\rrelay = [\orr{1}, \dots ,\orr{T}]$, where $\ors{i} = \frac{1}{2} \log{(1 + \ops{i})}$ and $\orr{i} = \frac{1}{2} \log{(1 + \opr{i})}$, for $i = 1,\dots,T$. First, we note that the optimal rate allocations of both the source and the relay are monotone non-decreasing sequences by \cite[Lemmas 1 and 4]{tcom-submit}, i.e., $\ors{i} \leq \ors{i+1}$ and $\orr{i} \leq \orr{i+1}$, for $i=1,\dots,T$. Second, we note the data causality constraint at the relay $\sum_{i=1}^k \orr{i} \leq \sum_{i=1}^k \ors{i}$, for all $k < T$, and the equality $\sum_{i=1}^T \orr{i} = \sum_{i=1}^T \ors{i}$ by Lemma~\ref{opt_cond1}. These imply that $\rsource$ is majorized by $\rrelay$, which is denoted by $\rsource \preceq \rrelay$; see \cite[Definition 1.A.1]{majorbook}. Since $\ops{i} = 2^{2 \ors{i}}-1$ and $g(x)=2^{2x}-1$ is strictly convex, $\sum_{i=1}^T \ops{i} = \sum_{i=1}^T 2^{2 \ors{i}}-1$ is a strictly Schur convex function of $\rsource$  \cite[Proposition 3.C.1]{majorbook}. Then, since $\rsource \preceq \rrelay$, we have that $\sum_{i=1}^T \ops{i} = \sum_{i=1}^T 2^{2 \ors{i}}-1 \leq \sum_{i=1}^T 2^{2 \orr{i}}-1 = \sum_{i=1}^T \opr{i}$ \cite[Proposition 4.B.1]{majorbook}. Moreover, due to the strict convexity of $g(x)$, and the resulting strict Schur convexity, equality is possible only when $\ors{i}=\orr{i}$ for all $i$.
\end{proof}

An immediate application of Lemma \ref{opt_cond4} is that if $\sum_{i=1}^T \er{i} < \sum_{i=1}^T \es{i}$, i.e., if the total energy of the relay is less than the total energy of the source, then the relay cannot forward the source data stream only with its own energy. In this case, we must have $\ods{i} \neq 0$ for some $i$, i.e., some energy transfer is strictly needed. We state this in the following lemma.
\begin{lemma}
\label{opt_cond5}
If the data buffer of the relay is empty at some slot $k$, $k \leq T$, then $\sum_{i=1}^k \ops{i} \leq \sum_{i=1}^k \opr{i}$, and with equality only when $\ops{i}=\opr{i}$ for all $i=1,\ldots,k$.
\end{lemma}

\begin{proof}
If the data buffer of the relay is empty at some slot $k$, $k \leq T$, then we must have $\sum_{i=1}^k \ors{i} = \sum_{i=1}^k \orr{i}$. Together with the data causality constraints at the relay $\sum_{i=1}^{\tilde{k}} \orr{i} \leq \sum_{i=1}^{\tilde{k}} \ors{i}$, for $\tilde{k}=1,\ldots,k-1$, we conclude that the subvector $\rsource_{k} = [\ors{1}, \dots, \ors{k}]$ is majorized by the subvector $\rrelay_{k} = [\ors{1}, \dots, \ors{k}]$, i.e., $\rsource_k \preceq \rrelay_k$. Then, $\sum_{i=1}^k \ops{i} = \sum_{i=1}^k 2^{2 \ors{i}}-1  \leq \sum_{i=1}^k 2^{2 \orr{i}}-1 = \sum_{i=1}^k \opr{i}$, and with equality iff $\rsource_k = \rrelay_k$ due to the strict Schur convexity.
\end{proof}

Necessary conditions in Lemmas~\ref{opt_cond1} through \ref{opt_cond5} do not provide detailed structural properties for the optimal policy for an algorithmic solution. In the next sections, we consider specific scenarios to gain insight on the optimal policy. In particular, we examine cases that correspond to practically interesting settings, such as the case of only one of the nodes harvesting energy.

\subsection{Specific Scenario: Relay Energy Higher at the Beginning Lower at the End}

We consider the scenario where the relay energy arrival profile is higher at the beginning, intersects the energy arrival profile of the source once, and remains lower until the end of the communication, as shown in Fig.~\ref{spec1}. In particular, we assume that there exists $\tilde{i} \in [0, T]$ such that $\sum_{k=1}^i \er{k} \geq \sum_{k=1}^i \es{k}$, for all $i=1,\ldots,\tilde{i}$, and $\sum_{k=1}^i \er{k} \leq \sum_{k=1}^i \es{k}$, for all $i=\tilde{i}+1,\ldots,T$. In Fig.~\ref{spec1}, $\tilde{i} = 3$. We note that this case also covers the setting where the relay is not energy harvesting, and only the source harvests energy during the communication session.

For this case, we propose the following solution. Form a new energy arrival profile as: $\min\{ \sum_{k=1}^i\frac{ \er{k} + \alpha\es{k}}{\alpha +1}, \sum_{k=1}^i\es{k}\}$ as shown in Fig.~\ref{spec1}, and maximize the throughput with respect to this profile. In particular, use $\sum_{k=1}^i\es{k}$ for $i=1,\ldots,\tilde{i}$, and $\sum_{k=1}^i\frac{\er{k} + \alpha \es{k}}{\alpha +1}$ for $i=\tilde{i} + 1,\ldots,T$; and perform energy transfer only at slots $\tilde{i}+1,\ldots,T$. The resulting power sequences are matched for the source and the relay. More specifically, we propose
\begin{align}
\ops{i} = \opr{i} = \frac{ \min \left\{ \frac{ \sum_{j=n_{i-1}}^{n_i} \er{j} + \alpha\es{j}}{\alpha +1}, \sum_{j=n_{i-1}}^{n_i}\es{j}\right\} }{n_i - n_{i-1}} \label{prop1}
\end{align}
where
\begin{align}
n_{i} = \arg\min_{n_{i-1}\leq k \leq T} \left\{ \frac{ \min\{\sum_{j=n_{i-1}}^{k} \frac{ \er{j} + \alpha\es{j}}{\alpha +1},\sum_{j=n_{i-1}}^{k}\es{j}\}}{k - n_{i-1}} \right\}
\label{prop2}
\end{align}

We next show that there exist $\lambda_i, \mu_i, \eta_i, \rho_i \geq 0$ that satisfy (\ref{KKT_relay})-(\ref{comp_slack}) and yield the solution in (\ref{prop1})-(\ref{prop2}) via (\ref{relay_soln})-(\ref{transfer_soln}). In particular, $\rho_i = 0$ and $\eta_i = \frac{\mu_i}{\alpha}$ for $i=\tilde{i}+1,\ldots,T$. Since $\alpha \sum_{k=i}^T \eta_k = \sum_{k=i}^T \mu_k$ for all $i=\tilde{i}+1\ldots,T$, we have from (\ref{relay_soln}) and (\ref{source_soln})
\begin{align}
\opr{i} + \alpha \ops{i} = \frac{1}{\sum_{k=i}^T \eta_k} - (1+\alpha), \quad i = \tilde{i}+1,\dots,T
\end{align}
Hence, $\opr{i}=\frac{1}{(1 + \alpha) \sum_{k=i}^T \eta_k} -  1$, which implies that $\lambda_T = \frac{1}{1+\alpha}$ and $\lambda_i=0$ for $i=\tilde{i}+1,\ldots,T-1$. Moreover, $\eta_i=\frac{\mu_i}{\alpha} >0$ whenever $\sum_{k=1}^i\frac{\er{k} + \alpha \es{k}}{\alpha +1}$ is active for some $i=\tilde{i} + 1,\ldots,T$. As in \cite{ozel11,wless-submit}, we can show that such $\eta_i=\frac{\mu_i}{\alpha}$ that yield the power sequence in (\ref{prop1})-(\ref{prop2}) are uniquely found for $i=\tilde{i} + 1,\ldots,T$.

\begin{figure}[t]
\begin{center}
\includegraphics[width=0.7\linewidth]{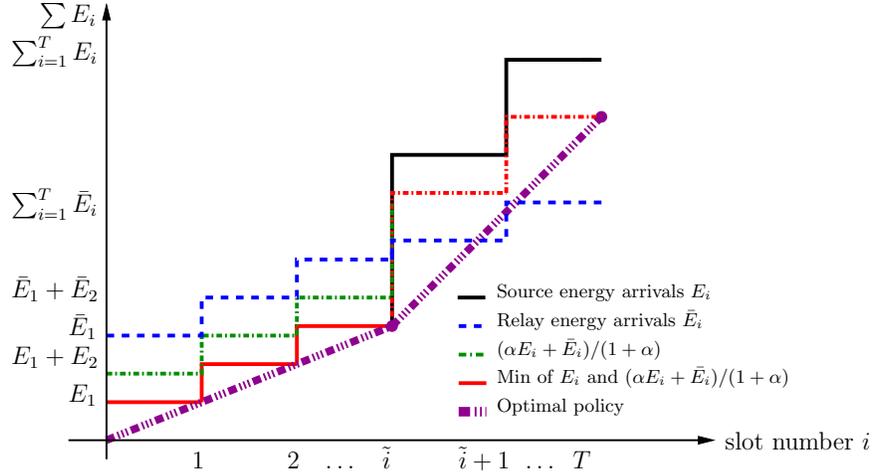}
\end{center}
\caption{Optimal power sequence and energy transfer when the relay energy profile is higher at the beginning and lower at the end with crossing only once.}
\label{spec1}
\vspace*{-0.1in}
\end{figure}

It remains to find the Lagrange multipliers for $i=1,\ldots,\tilde{i}$. We observe that $\eta_i = 0$ and $\rho_i = \sum_{k=i}^{\tilde{i}} \mu_k$ for $i=1,\ldots,\tilde{i}$. That is, the relay power constraint is not active in the first $\tilde{i}$ slots, i.e., $\sum_{k=1}^i \opr{k} < \sum_{k=1}^i \er{k}$, $i=1,\ldots,\tilde{i}$. To justify this claim, we note that since $\ops{i}=\opr{i}$ for $i=\tilde{i}+1,\ldots,T$, we have $\sum_{i=1}^{\tilde{i}} \frac{1}{2} \log{(1+ \ops{i})} =  \sum_{i=1}^{\tilde{i}} \frac{1}{2}\log{(1+ \opr{i})}$. By Lemma~\ref{opt_cond5}, selecting $\ps{i}=\pr{i}$ for $i=1,\ldots,\tilde{i}$ is the minimum energy consuming policy at the relay. Since by assumption $\sum_{k=1}^i \ps{k} \leq \sum_{k=1}^i \pr{k}$ for $i=1,\ldots,\tilde{i}$, $\ps{i}=\pr{i}$ is feasible and hence optimal, which in turn implies that $\sum_{k=1}^i \opr{k} < \sum_{k=1}^i \er{k}$ for $i=1,\ldots,\tilde{i}$. As a consequence, $\sum_{k=i}^T \eta_k = \sum_{k=\tilde{i}+1}^T \eta_k$, i.e., constant for all $i=1,\ldots,\tilde{i}$. As $\opr{i} \leq \opr{\tilde{i}+1}$, we can specify $0 \leq \lambda_i \leq \frac{1}{1+\alpha}$ recursively, with $\lambda_i >0$ only when $\sum_{k=1}^i\es{k}$ constraint is active, as follows
\begin{align}
\lambda_i = 1 - \opr{i}\sum_{k=\tilde{i}+1}^T \eta_k - \sum_{k=i+1}^{T} \lambda_k
\end{align}
Moreover, $\mu_i >0$ for slots where $\sum_{k=1}^i\es{k}$ constraint is active and $\mu_i = \frac{\sum_{k=i}^T \lambda_k}{\ops{i}} - \sum_{k=i+1}^T \mu_k$. Note that if $\ods{i} \neq 0$ for some $i$, the optimal source and relay power sequences are unique while there may exist infinitely many $\ods{i}$ that yield the same optimal power levels.

A particular case covered is when only the source has energy replenishments and the relay has all its energy available initially, i.e., $\er{1}> 0$ and $\er{i} = 0$ for $i>1$. If $\er{1} > \sum_{i=1}^T \es{i}$, the relay can forward all the bits sent from the source and the optimal policy is trivial. If $\er{1} < \sum_{i=1}^T \es{i}$, the optimal policy is obtained by forming a common energy profile via energy transfer and matching the power and rate sequences. Another special case is when $\tilde{i}=0$, i.e., when $\er{i}<\es{i}$ for all $i$. In this case, $\min\{ \sum_{k=1}^i\frac{\er{k} + \alpha\es{k}}{\alpha + 1},\sum_{k=1}^i\es{k}\} = \sum_{k=1}^i\frac{\alpha\er{k} + \es{k}}{\alpha + 1}$ for all $i$ and matching the relay and source power sequences is optimal with $\delta_i^* = \es{i} - \frac{\er{i} + \alpha\es{i}}{\alpha + 1}$. When $\tilde{i} = T$, we have $\er{i}>\es{i}$, $\forall i$. The source optimizes the throughput according to $\{\es{i}\}_{i=1}^{T}$ and the relay power is matched with the source.

\begin{figure}[t]
\begin{center}
\includegraphics[width=0.56\linewidth]{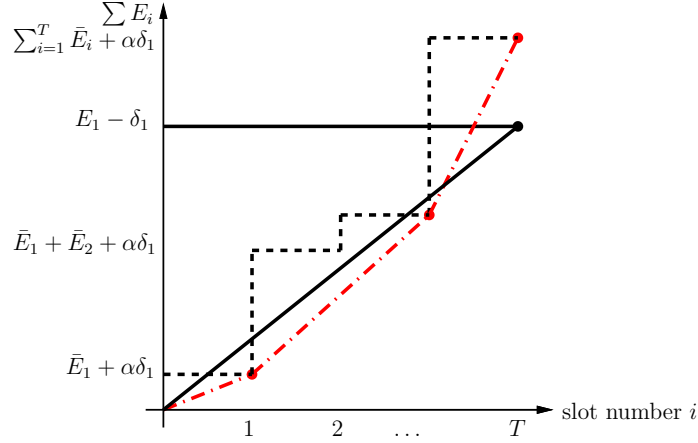}
\end{center}
\caption{Optimal power sequences and energy transfer when the source energy is available at the beginning.}
\label{spec2}
\vspace*{-0.1in}
\end{figure}

\subsection{Specific Scenario: Source Energy Available at the Beginning}

We consider the scenario where the source has all of its energy available at the beginning (i.e., $\es{1}>0$ only), and the relay harvests energy throughout the communication. Let the relay energy profile not be  satisfactory to forward the optimal source data stream which has constant rate $\frac{1}{2} \log{(1 + \frac{\es{1}}{T})}$. Assume $\delta_i \neq 0$ for some $i$. Since the source is not energy harvesting, the total energy of the source will then be $E_1 - \delta_i$ yielding an optimal transmission power of $\frac{E_1 - \delta_i}{T}$. Hence, the throughput of the source is independent of the slot index $i$ the energy is transferred. However, transferring the energy at slot $j<i$ can only increase the relay transmit powers after that slot; therefore, energy transfer has to be performed as early as possible, i.e., at the first slot. Hence, the jointly optimal policy is $\ods{1}\neq 0$ and $\ods{i}=0$ for the remaining slots as shown in Fig. \ref{spec2}. Note that the power sequences of the source and the relay are not matched. $\ods{1}$ is found by solving a fixed point equation as:
\begin{align}
f(\er{1}+\ods{1},\er{2},\ldots,\er{T}) = \frac{T}{2}\log\left(1+\frac{\es{1}-\ods{1}}{T}\right)
\end{align}
where $f(\er{1},\er{2},\ldots,\er{T})$ is the maximum number of bits corresponding to the energy arrival sequence $\er{1},\er{2},\ldots,\er{T}$.

\section{Gaussian Two-Way Channel with One-Way Energy Transfer}
\label{twcmodel}

In this section, we consider a two-way channel as shown in Fig.~\ref{sysmodtwc}. The two queues at the nodes are the data and energy queues. The energies that arrive at the nodes are saved in the corresponding energy queues. The data queues of both users always carry some data packets. The physical layer is a memoryless Gaussian two-way channel \cite{tesunhan} where the channel inputs and outputs are $x_1$, $x_2$ and $y_1$, $y_2$, respectively. The input-output relations are $y_1=x_1 + x_2 + n_1$ and $y_2=x_1 + x_2 + n_2$ where $n_1$ and $n_2$ are independent Gaussian noises with zero-mean and unit-variance. In slot $t$, the first and second users harvest energy in amounts $\es{t}$ and $\er{t}$, respectively. There is a separate wireless energy transfer unit at the first user, that transfers energy from the first user to the second user with efficiency $0 \leq \alpha \leq 1$. The power policy of user $1$ is composed of the sequences $\ps{i}$ and $\ds{i}$, and the power policy of user $2$ is the sequence $\pr{i}$.

For the Gaussian two-way channel with individual power constraints $P_1$ and $P_2$, rate pairs $(R_1,R_2)$ with $R_1 \leq \frac{1}{2} \log{(1 + P_1)}, R_2 \leq \frac{1}{2} \log{(1 + P_2)}$ are achievable \cite{tesunhan}. For a fixed energy transfer vector $\deltav$, and feasible power control policies $\psource$ and $\prelay$, the set of achievable rates is:
\begin{align}
\mathcal{C}_{\deltav} (\psource,\prelay) =  \Big \{ (R_1, R_2): \ R_1 \leq \sum_{i=1}^T \frac{1}{2} \log{(1+\ps{i})}, \ \ R_2 \leq \sum_{i=1}^T \frac{1}{2} \log{(1+\pr{i})}  \Big\} \label{capdelta}
\end{align}
The notation shows the dependence of the region on the energy transfer vector $\deltav$. This region is shown in Fig.~\ref{twowaycap} for different values of $\deltav$. Each of these regions are rectangles of the form $R_i \leq C_i$ where $C_i$ is the maximum throughput achieved for user $i$ found by maximizing (\ref{capdelta}) constrained to the feasibility constraints $\mathcal{F}$. As $\deltav$ is increased, energy is transferred from user 1 to user 2 therefore $C_1$ decreases while $C_2$ increases. By taking the union of the regions over all possible energy transfer vectors and power policies for the users, we obtain the capacity region of the Gaussian two-way channel as:
\begin{equation}
\mathcal{C} (\esource,\erelay) = \bigcup_{(\deltav,\psource,\prelay) \in \mathcal{F}}
\mathcal{C}_{\deltav} (\psource,\prelay) \label{twc-cap-region}
\end{equation}
We determine the capacity region of the Gaussian two-way channel in the next section, by solving weighted rate maximization problems which trace the boundary of the capacity region.

\begin{figure}
\begin{center}
\includegraphics[width=0.6\linewidth]{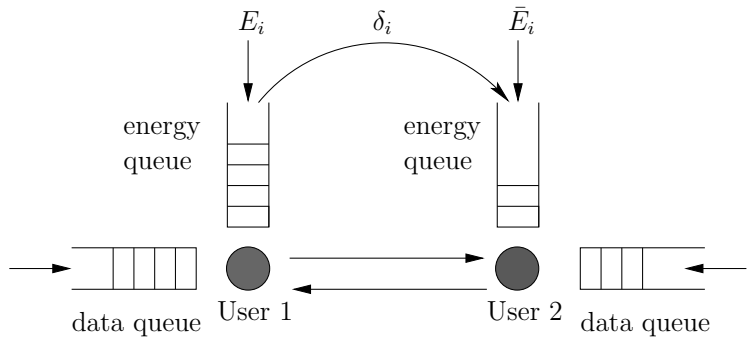}
\end{center}
\caption{Two-way channel with one-way energy transfer.}
\label{sysmodtwc}
\vspace*{-0.1in}
\end{figure}

\section{Capacity Region of the Gaussian Two-Way Channel}
\label{twoway}

In this section, we characterize the capacity region as well as the optimal power allocation and energy transfer policies. We start by noting that the capacity region is convex in the following lemma. The proof of this lemma is provided in Appendix~\ref{pf1}.
\begin{lemma}
\label{convexlemma}
$\mycap$ is a convex region.
\end{lemma}

Since $\mycap$ is convex, each boundary point can be found by solving the following weighted rate maximization problem:
\begin{eqnarray}
\nonumber \max_{\pr{i}, \, \ps{i}, \, \ds{i}} & &  \sum_{i=1}^T \hp_1 \frac{1}{2} \log{(1+\ps{i})} + \hp_2 \frac{1}{2} \log{(1+\pr{i})} \\ \mbox{s.t.} & & (\deltav,\psource,\prelay) \in \mathcal{F} \label{opt_prob1}
\end{eqnarray}
The problem in (\ref{opt_prob1}) is a convex optimization problem as the objective function is concave and the feasible set is a convex set \cite{boyd}. We write the Lagrangian function for (\ref{opt_prob1}) as:
\begin{align}
\nonumber \mathcal{L} =&- \sum_{i=1}^T \left[ \hp_1 \log{(1+\ps{i})} + \hp_2 \log{(1+\pr{i})} \right] + \sum_{k=1}^T \mu_k \left(\sum_{i=1}^k \ps{i} - (\es{i} - \delta_i) \right) \\ &+ \sum_{k=1}^T \eta_k \left(\sum_{i=1}^k \pr{i} - (\er{i} + \alpha \delta_i) \right) - \sum_{k=1}^T \rho_k \delta_k
\label{lagrangian_relay}
\end{align}
where we do not include constraints $\sum_{i=1}^k \delta_i \leq \sum_{i=1}^k \es{i}$ since these constraints can never be satisfied with equality in the optimal policy for the Gaussian two-way channel for $\hp_1,\hp_2 > 0$, as that would require $\ps{i}=0$ for some $i$, which is suboptimal. Zero powers are optimal only in the degenerate case $\theta_1=0$, which corresponds to maximizing $R_2$, whose solution is point 3 in Fig.~\ref{twowaycap}. The KKT conditions for this problem are:
\begin{align}
-\frac{\hp_1}{1 + \ps{i}} + \sum_{k=i}^T \mu_k = 0,  \quad \forall i\label{KKT_source_TWC}\\
-\frac{\hp_2}{1 + \pr{i}} + \sum_{k=i}^T \eta_k = 0, \quad \forall i \label{KKT_relay_TWC}\\
\sum_{k=i}^T \mu_k - \alpha \sum_{k=i}^T \eta_k - \rho_i = 0,  \quad \forall i \label{KKT_energy_TWC}
\end{align}
with the additional complementary slackness conditions as:
\begin{align}
\mu_k  \left(\sum_{i=1}^k \ps{i} - (\es{i} - \ds{i}) \right) & = 0, \quad \forall k \label{comp_slack1} \\
\eta_k  \left(\sum_{i=1}^k \pr{i} - (\er{i} + \alpha \ds{i}) \right) & = 0, \quad \forall k \label{comp_slack2} \\
\rho_k \ds{k} & = 0, \quad \forall k \label{comp_slack3}
\end{align}
From (\ref{KKT_source_TWC}), (\ref{KKT_relay_TWC}) and (\ref{KKT_energy_TWC}) we get:
\begin{align}
\ps{i} &=  \frac{\hp_1}{\sum_{k=i}^T \mu_k} -1,
\quad \forall i \label{source_soln_TWC} \\
\pr{i} &=   \frac{\hp_2}{\sum_{k=i}^T \eta_k} -1,
\quad \forall i \label{relay_soln_TWC} \\
\rho_i &= \sum_{k=i}^T \mu_k - \alpha \sum_{k=i}^T \eta_k, \quad \forall i \label{transfer_soln_TWC}
\end{align}

\begin{figure}[t]
\begin{center}
\includegraphics[scale=0.85]{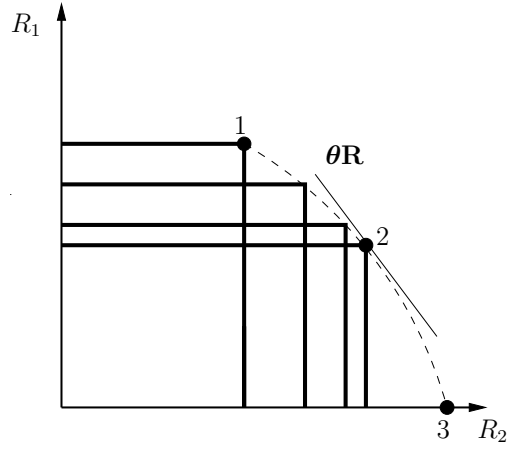}
\end{center}
\caption{Capacity region of the Gaussian two-way channel.}
\label{twowaycap}
\vspace*{-0.1in}
\end{figure}

We will give the solution for general $\hp_1,\hp_2 >0$ in the sequel. Before that, we note that in the extreme case when $\hp_2 = 0$, the problem reduces to maximizing the first user's throughput only and hence any energy transfer is strictly sub-optimal, i.e.,  $\deltav = \mathbf{0}$ is optimal. This corresponds to point 1 in Fig.~\ref{twowaycap}. Similarly, when $\hp_1 = 0$, the problem reduces to maximizing the second user's throughput only and the first user must transfer all of its energy to the second user, i.e., $\deltav = \esource$ is optimal. This corresponds to point 3 in Fig.~\ref{twowaycap}. When $\hp_1,\hp_2>0$, we obtain the points between points 1 and 3 in Fig.~\ref{twowaycap}. In this case, for a given energy transfer profile $\delta_1,\dots,\delta_T$, the optimization problem can be separated into two optimization problems, each only in terms of the power control policy of the corresponding user. For fixed $\deltav$, the optimal power policies of the two users can be found by \cite{tcom-submit}.

Next, we provide the necessary optimality condition for a non-zero energy transfer.
\begin{lemma}
\label{ratiolemma}
For the optimal power sequences $\ops{i},\opr{i}$ and energy transfer sequence $\ods{i}$, if $\ods{i} \neq 0$ for a slot $i$, then
\begin{equation}
\frac{1+\ops{i}}{1+\opr{i}} = \frac{\hp_1}{\hp_2 \alpha} \label{ratio}
\end{equation}
\end{lemma}
\begin{Proof}
From (\ref{source_soln_TWC}), (\ref{relay_soln_TWC}) and (\ref{transfer_soln_TWC}), we have
\begin{equation}
\frac{1+\ops{i}}{1+\opr{i}} = \frac{\hp_1 \sum_{k=i}^T \eta_k}{\hp_2 (\alpha \sum_{k=i}^T \eta_k + \rho_i) }
\end{equation}
If there is a non-zero energy transfer, $\ods{i} \neq 0$, we have from (\ref{comp_slack3}), $\rho_i = 0$. Therefore, (\ref{ratio}) must be satisfied if $\ods{i}\neq 0$.
\end{Proof}

In order to devise an algorithmic solution, we apply a change of variable $\prnew{i} = \frac{\pr{i}}{\alpha}$ and re-write the optimization problem in terms of $\ps{i},\prnew{i},\ds{i}$ as follows:
\begin{eqnarray}
\nonumber \max_{\prnew{i}, \, \ps{i}, \, \ds{i}} & &  \sum_{i=1}^T \hp_1 \frac{1}{2} \log{(1+ \ps{i})} + \hp_2 \frac{1}{2} \log{(1+ \alpha \prnew{i})}  \\
\nonumber \mbox{s.t.} & & \sum_{i=1}^k \ps{i} \leq \sum_{i=1}^k (\es{i} - \delta_i), \quad \forall k\\
\nonumber & & \sum_{i=1}^k \prnew{i} \leq \sum_{i=1}^k \left(\frac{\er{i}}{\alpha} + \delta_i\right), \quad \forall k \\
 & & \sum_{i=1}^k \delta_i \leq \sum_{i=1}^k \es{i}, \quad \forall k \label{opt_prob2}
\end{eqnarray}
The optimal power allocation for this problem is:
\begin{align}
\ops{i} &= \hp_1 \wls{i} - 1, \quad \forall i \label{newlevel1} \\
\oprnew{i} &= \hp_2 \wlnew{i} - \frac{1}{\alpha}, \quad \forall i \label{newlevel2}
\end{align}
where $\wls{i}$ and $\wlnew{i}$ in slot $i$ are defined by
\begin{align}
\wls{i} = \frac{1}{\sum_{k=i}^T \mu_k} \quad \mbox{and} \quad
\wlnew{i} = \frac{1}{\sum_{k=i}^T \eta_k} \label{wlr_define}
\end{align}

The power level expressions in (\ref{newlevel1})-(\ref{newlevel2}) lead to a directional water-filling interpretation \cite{ozel11}. In particular, we note that energy has to be jointly allocated in time and user dimensions together. This calls for a \textit{two-dimensional directional water-filling} algorithm where energy is allowed to flow in two dimensions, from left to right (in time) and from up to down (among users). We, next, explain this algorithm.

\subsection{Two-Dimensional Directional Water-filling Algorithm}

We utilize \textit{right permeable taps} to account for energy which will be used in the future and \textit{down permeable taps} to account for energy that will be transferred from user $1$ to user $2$; see Fig.~\ref{twodirwfmetertaps}. We note from the KKT conditions that $\wls{i} = \wlnew{i}$ in slots where there is non-zero energy transfer. Note that in the original problem, this implies that if some energy is transferred, then the power levels in that slot need to satisfy (\ref{ratio}). The base levels for users 1 and 2 are $1$ and $\frac{1}{\alpha}$, respectively. Moreover, to facilitate the water flow interpretation, we scale the energy arrivals of user 2 by $\frac{1}{\alpha}$ as seen in (\ref{opt_prob2}). If the resulting water levels are higher for user $1$ or not monotonically increasing in time for both users, then water has to flow until the levels are balanced.

While finding the balanced water levels, the two dimensions of the water flow (i.e., in time and among users) are coupled and therefore it is not easy to determine beforehand which taps will be open or closed in the optimal solution. In particular, the water flow of user $2$ from time slot $i$ to time slot $i+j$, $j>0$, may become redundant if some energy is transferred from user $1$ in time slot $i+j$. To circumvent this difficulty, we let each tap (right/down permeable) have a \textit{meter} measuring the water that has already passed through it and we allow that tap to let the water flow back if an update in the allocation necessitates it. This way, we keep track of the source of the energy and whether it is transferred to future time slots or to the other user. First, we fill energy into the slots with all taps closed. Then, we open only the right permeable taps and perform directional water-filling (over time) for both users individually \cite{ozel11}. Then, we open the down taps one by one in a backward fashion. If water flows down through a tap, the amount is measured by the meter. Water levels in the slots connected by the bi-directional horizontal taps need to be equal. Whenever water flows down through a down permeable tap, the water levels must satisfy the proportionality relationship in (\ref{ratio}). When the water levels are properly balanced, the optimal solution is obtained.

An example run of the algorithm is given in Fig.~\ref{twodirwfmetertaps} for $\theta_1=\theta_2$ and $\alpha=1$. Initially, we open the right permeable taps and the water levels are equalized for the first user. Then, we open the down permeable taps. In the second slot there is no need for energy transfer because $\frac{E_1 + E_2}{2} < \er{2}$. In the first slot there will be some non-zero energy transfer since $\frac{E_1 + E_2}{2} > \er{1}$, and some water flows through the first down permeable tap. Since user 1's right permeable tap has a positive meter at that point, some water is allowed to flow from right to left thereby equalizing the water levels of user 1's first and second slots and user 2's first slot.

\begin{figure}[t]
\begin{center}
\includegraphics[scale=0.5]{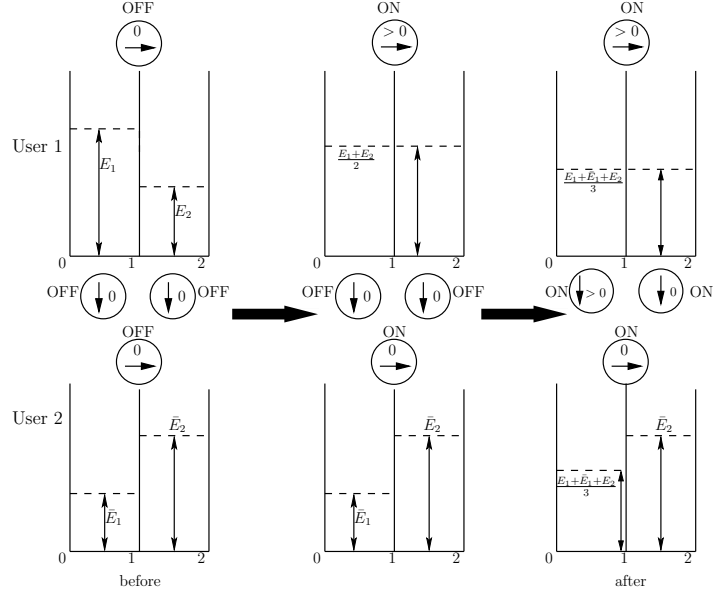}
\end{center}
\caption{Two-dimensional directional water-filling with right/down permeable meter taps for a two slot system.}
\label{twodirwfmetertaps}
\vspace*{-0.1in}
\end{figure}

\subsection{A Specific Run of the Algorithm}

In order to show more specifically how the algorithm runs and further justify the particular sequence of steps followed in the two-dimensional water-filling algorithm, we next provide a numerical example where $\esource =[0,12,0]$ mJ, $\erelay = [6,6,0]$ mJ and $\alpha = 1$. $\htap{1}{i}$, $\htap{2}{i}$ denote the horizontal taps of the first and second users connecting the $i$th and $i+1$st slots, $\vtap{i}$ denotes the $i$th vertical tap. The optimal solution is $\psource = [0,4.8,4.8]$ and $\prelay = [4.8,4.8,4.8]$, which is obtained by spreading the energy as equally as possible in two dimensions among the users and time slots, subject to energy causality. In order to understand why the particular order of tap openings is optimal, we next consider two sub-optimal orderings of tap openings.

Assume that we open the horizontal taps first and keep the vertical taps closed. This yields the transient water levels $\psource = [0,6,6]$ and $\prelay = [4,4,4]$. Now, if we open the vertical taps, water is transferred in the second and third slots and the balanced final levels are $\psource = [0,5,5]$ and $\prelay = [4,5,5]$. This profile is not optimal since the second user changes its power level when the battery is non-empty, violating \cite[Lemma 2]{tcom-submit}.

Now, assume that we open the vertical taps first and keep the horizontal taps closed. Energy is transferred in the second slot and the new transient water levels will be $\psource = [0,9,0]$ and $\prelay = [6,9,0]$. Then, when we open the horizontal taps, we will have $\psource = [0, 4.5, 4.5]$ and $\prelay = [5,5,5]$. This profile is not optimal either, as after energy transfer, the source power level is less than the relay power level, violating Lemma~\ref{ratiolemma}.

We now show how the two-dimensional directional water-filling algorithm works. First, we open the horizontal taps to get $\psource = [0,6,6]$ and $\prelay = [4,4,4]$ with the water meters reading $[0,6]$ and $[2,2]$. Recall that the taps with positive meter readings allow bi-directional energy transfer. Next, we open the vertical taps in a backward fashion. Once $\vtap{3}$ is opened, water flows to the second user and since $\htap{2}{1}$, $\htap{2}{2}$ are bi-directional it starts to fill all the slots of the second user. A balance is established when $\psource = [0,4.8,4.8]$ and $\prelay = [4.8,4.8,4.8]$, which is the optimal solution.

\section{Multiple Access Channel with One-Way Energy Transfer}
\label{macmodel}

In this section, we consider the multiple access channel scenario shown in Fig.~\ref{macmodelfig}. In the multiple access channel, the received signal is $y = x_1 + x_2 + n$ where $x_1$ and $x_2$ are signals of user 1 and user 2, respectively, and $n$ is a Gaussian noise with zero-mean and unit-variance. For the Gaussian two-user multiple access channel with individual power constraints $P_1$ and $P_2$, rate pairs $(R_1,R_2)$ with $R_1 \leq \frac{1}{2} \log{(1 + P_1)}, R_2 \leq \frac{1}{2} \log{(1 + P_2)}$, $R_1 + R_2 \leq \frac{1}{2} \log{(1 + P_1 + P_2)}$ are achievable \cite{Cover06}. For a fixed energy transfer vector $\deltav$, and feasible power control policies $\psource$ and $\prelay$, the set of achievable rates is a pentagon defined as \cite{jing12jcn}:
\begin{align}\nonumber
\mathcal{C}_{\deltav} (\psource,\prelay) =  \Big\{ (R_1, R_2): R_1 &\leq \sum_{i=1}^T \frac{1}{2} \log{(1+\ps{i})}, \\ \nonumber R_2 &\leq \sum_{i=1}^T \frac{1}{2} \log{(1+\pr{i})}, \\  R_1 + R_2 &\leq \sum_{i=1}^T \frac{1}{2} \log{(1+\pr{i} + \ps{i})} \label{capdeltamac} \Big\}
\end{align}
For each feasible $(\psource,\prelay,\deltav)$, the region is a pentagon. We obtain the capacity region by taking the union of these regions over all feasible power allocations and energy transfer profiles:
\begin{equation}
\mathcal{C} (\esource,\erelay) = \bigcup_{(\deltav,\psource,\prelay) \in \mathcal{F}}
\mathcal{C}_{\deltav} (\psource,\prelay)
\end{equation}
We determine the capacity region of the Gaussian multiple access channel in the next section.

\begin{figure}[t]
\begin{center}
\includegraphics[scale=0.76]{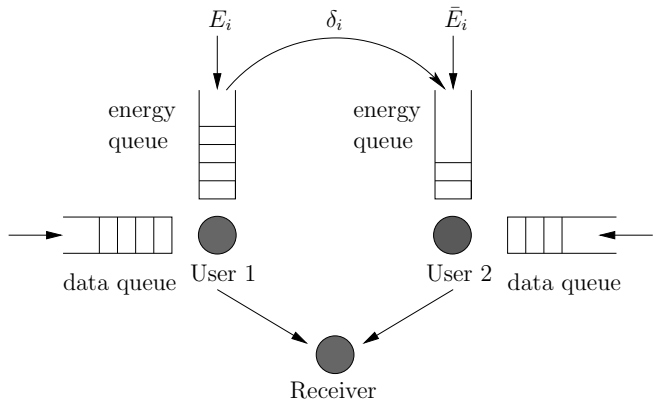}
\end{center}
\caption{Multiple access channel with one-way energy transfer.}
\label{macmodelfig}
\vspace*{-0.1in}
\end{figure}

\section{Capacity Region of the Gaussian Multiple Access Channel}
\label{macsection}

In this section, we characterize the capacity region as well as the optimal power allocation and energy transfer policies. First, we note in the following lemma that the capacity region is convex. We prove this lemma in Appendix~\ref{pf2}.

\begin{lemma}
\label{convexlemmamac}
$\mathcal{C} (\esource,\erelay)$ is a convex region.
\end{lemma}

Since the region is convex, each boundary point is a solution to $\max_{\textbf{R} \in \mathcal{C}^M} \boldsymbol{\hp} \textbf{R}$ \cite{tse_mac} for some $\boldsymbol{\hp}=[\hp_1,\hp_2]$. We examine two cases separately, $\hp_1 \geq \hp_2$ and  $\hp_1 < \hp_2$.

\subsection{$\hp_1 \geq \hp_2$}
We show that when $\hp_1 \geq \hp_2$, no energy transfer from user 1 to user 2 is needed. Note that as $\hp_1 \geq \hp_2$, the boundary points between $1$, $2$ and $3$ in Fig.~\ref{maccap} are found by solving the problem:
\begin{align}
\nonumber \max_{\pr{i}, \, \ps{i}, \; \ds{i}} & \, \sum_{i=1}^T  (\hp_1-\hp_2)\frac{1}{2} \log{(1+\ps{i})} + \hp_2 \frac{1}{2} \log{(1+\pr{i}+\ps{i})} \\
\mbox{s.t.} & \; (\deltav,\psource,\prelay) \in \mathcal{F} \label{opt_prob1m}
\end{align}
The problem in (\ref{opt_prob1m}) is a convex optimization problem as the objective function is concave and the feasible set is a convex set \cite{boyd}. We write the Lagrangian function for (\ref{opt_prob1m}) as:
\begin{align}
\nonumber \mathcal{L} =&- \sum_{i=1}^T \left[ (\hp_1-\hp_2) \log{(1+\ps{i})} + \hp_2 \log{(1+\pr{i}+\ps{i})} \right] + \sum_{k=1}^T \mu_k \left(\sum_{i=1}^k \ps{i} - (\es{i} - \delta_i) \right) \\ &+ \sum_{k=1}^T \eta_k \left(\sum_{i=1}^k \pr{i} - (\er{i} + \alpha \delta_i) \right) + \sum_{k=1}^T \gamma_k \left(\sum_{i=1}^k \delta_i - \es{i} \right) - \sum_{k=1}^T \rho_k \delta_k
\label{lagrangian_mac}
\end{align}
where we now included the constraints $\sum_{i=1}^k \delta_i \leq \sum_{i=1}^k \es{i}$. The KKT conditions are:
\begin{align}
-\frac{\hp_1 - \hp_2}{1 + \ps{i}} - \frac{\hp_2}{1 + \ps{i} + \pr{i}} + \sum_{k=i}^T \mu_k = 0,  \quad \forall i \label{KKT_sourcem}\\
-\frac{\hp_2}{1 + \ps{i} + \pr{i}} + \sum_{k=i}^T \eta_k = 0, \quad \forall i \label{KKT_relaym}\\
\sum_{k=i}^T \mu_k - \alpha \sum_{k=i}^T \eta_k + \sum_{k=i}^T \gamma_k  - \rho_i = 0,  \quad \forall i \label{KKT_energym}
\end{align}
Since $\hp_1 \geq \hp_2$, from (\ref{KKT_sourcem})-(\ref{KKT_relaym}), we have $\sum_{k=i}^T \mu_k \geq \sum_{k=i}^T \eta_k$, which is satisfied with equality iff $\hp_1=\hp_2$. This together with (\ref{KKT_energym}) implies that $\rho_i - \sum_{k=i}^T \gamma_k \geq 0$, which is satisfied with equality iff $\hp_1=\hp_2$ and $\alpha=1$. Therefore, unless we have exactly $\hp_1=\hp_2$ and $\alpha=1$, we must have $\rho_i > 0$ for all $i$. This together with the complementary slackness conditions $\rho_k \delta_k=0$ implies that we must have $\ds{i} = 0$ for all $i$, i.e., no energy transfer is needed. However, when $\theta_1=\theta_2$ and additionally if  $\alpha=1$, then there may exist multiple different optimal energy transfer profiles, including the one with no energy transfer.

\begin{figure}[t]
\begin{center}
\includegraphics[scale=0.85]{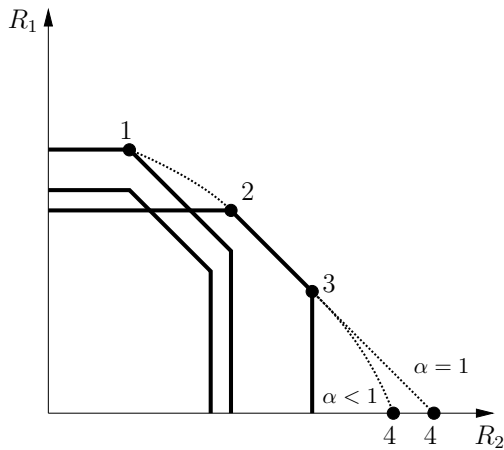}
\end{center}
\caption{Capacity region of the Gaussian multiple access channel for $\alpha=1$ and $\alpha<1$.}
\label{maccap}
\vspace*{-0.1in}
\end{figure}

Since energy transfer is not needed, optimal power control policies for the two users are the same as those in the energy harvesting multiple access channel with no energy transfer and can be found by the \textit{generalized backward directional water-filling algorithm} described in \cite{jing12jcn}. That is, the capacity region boundary from point $1$ to point $3$ in Fig.~\ref{maccap} is found by the algorithm in \cite{jing12jcn}. Specifically, for $\hp_1=\hp_2$, we have $\eta_k=\mu_k$ for all $k$ and the sum-rate optimal power policies are obtained by applying single-user directional water-filling algorithm to the sum of the energy profiles of the two users \cite{jing12jcn}.

\subsection{$\hp_1 < \hp_2$}

Here, we consider the remaining parts of the boundary, namely the points from point $3$ to point $4$ in Fig.~\ref{maccap}.  In this case, we need to solve the following optimization problem:
\begin{eqnarray}
\nonumber \max_{\pr{i}, \, \ps{i}, \, \ds{i}} & &   \sum_{i=1}^T (\hp_2-\hp_1)\log{(1+\pr{i})} + \hp_1 \log{(1+\pr{i}+\ps{i})} \\
\mbox{s.t.} & & (\deltav, \, \prelay, \, \psource) \in \mathcal{F} \label{opt_prob2m}
\end{eqnarray}
which is a convex optimization problem and the corresponding KKT conditions are:
\begin{align}
-\frac{\hp_1}{1 + \ps{i} + \pr{i}} + \sum_{k=i}^T \mu_k &= 0, \quad \forall i \label{KKT_relaym2}\\
-\frac{\hp_2 - \hp_1}{1 + \pr{i}} - \frac{\hp_1}{1 + \ps{i} + \pr{i}} + \sum_{k=i}^T \eta_k &= 0,  \quad \forall i
\label{KKT_sourcem2}\\
\sum_{k=i}^T \mu_k - \alpha \sum_{k=i}^T \eta_k + \sum_{k=i}^T \gamma_k - \rho_i &= 0,  \quad \forall i \label{KKT_energym2}
\end{align}
We do not have an analytical closed form solution for (\ref{KKT_relaym2})-(\ref{KKT_energym2}). Since (\ref{opt_prob2m}) is a convex optimization problem, standard numerical methods for convex optimization may be employed. We find that the solution of (\ref{opt_prob2m}) has a simple form in some special cases, which we investigate next.

When $\alpha=1$, we find that the optimal solution of (\ref{opt_prob2m}) requires all the energy of user $1$ transferred to user $2$. To verify this fact, we note from (\ref{KKT_relaym2})-(\ref{KKT_sourcem2}) that $\eta_T>\mu_T$, since $\hp_2>\hp_1$. Combining this with (\ref{KKT_energym2}), we obtain $ \gamma_T - \rho_T > 0$. Note that if $\sum_{i=1}^T \delta_i < \sum_{i=1}^T \es{i}$, then $\gamma_T=0$ and hence $\rho_T<0$, which is not possible. Thus, in the optimal solution, we must have $\sum_{i=1}^T \delta_i = \sum_{i=1}^T \es{i}$. Therefore, user 1 should not transmit any data, and instead should transfer all of its energy to user 2 by the end of $T$ slots. This policy corresponds to point $4$ in Fig.~\ref{maccap}. On the other hand, sum-rate optimal point, point 3, achieves the same throughput as point 4. This implies that when $\alpha=1$, points $2$, $3$ and $4$ in Fig.~\ref{maccap} lie on the $45^{o}$ line. In particular, the optimal throughput of user $2$, which is obtained by single-user throughput maximization subject to harvested energies of user 2 plus the harvested energies of user 1, coincides with the optimal sum-throughput.

When $\alpha<1$, points $3$ and $4$ in Fig.~\ref{maccap} are not on the same line. We observe that when $\frac{\hp_2}{\hp_1}$ is sufficiently large, user 1 transfers all of its energy to user 2. In order to verify this claim, we note that, if user 1 transfers some but not all of its energy at the end of $T$ slots, then $\gamma_T=0$. In this case, from (\ref{KKT_relaym2})-(\ref{KKT_energym2}) and as $\rho_T \geq 0$, we have
\begin{align}
\label{ff}
\frac{1+\pr{T}}{1+\pr{T}+\ps{T}} \geq \frac{\alpha(\hp_2-\hp_1)}{(1-\alpha) \hp_1}
\end{align}
Since $\frac{1+\pr{T}}{1+\pr{T}+\ps{T}}<1$, we conclude that if $\frac{\alpha(\hp_2-\hp_1)}{(1-\alpha) \hp_1}\geq 1$, then (\ref{ff}) cannot be satisfied which forces all of the energy of user 1 to be transferred to user 2 so that $\gamma_T>0$. Note that $\frac{\alpha(\hp_2-\hp_1)}{(1-\alpha) \hp_1}\geq 1$ is equivalent to $\frac{\hp_2}{\hp_1}\geq \frac{1}{\alpha}$. Hence, if $\frac{\hp_2}{\hp_1}\geq \frac{1}{\alpha}$, in the optimal solution, user 1 transfers all of its energy to user 2. This implies that the capacity region boundary intersects the horizontal line in Fig. \ref{maccap} with slope less than or equal to $\frac{1}{\alpha}$.

\section{Numerical Results}

In this section, we provide numerical examples for the multi-user settings studied and illustrate the resulting optimal policies. In all examples, we assume that the slot length is 1 second, noise spectral density is $N_0=10^{-19}$ W/Hz and the available bandwidth is 1 MHz. Moreover, path loss of each link in each model is set to 100 dB for convenience.

\subsection{Numerical Example for the Gaussian Two-Hop Relay Channel}

We first consider the two-hop relay channel with energy harvesting and energy transfer in Section~\ref{model}. In our first numerical study, the source and the relay have the energy arrival profiles $\esource = [2 ;3 ;5 ;4]$ mJ and $\erelay = [5 ;1 ;2 ;1]$ mJ, respectively, and the wireless energy transfer efficiency is $\alpha = 0.5$. We note that for these energy harvesting profiles the relay energy profile is higher at the beginning and lower at the end with crossing only once in the third slot. Therefore, the resulting optimal rate profiles are matched in the optimal policy. An optimal energy transfer vector is $\deltav =  [0 ;0 ;1.33 ;3.33]$ mJ and the resulting optimal power allocation vectors after the energy transfer are $\prelay = \psource = [2 ;3 ;4 ;6.33]$ mW. We note that while the optimal energy transfer profile is not unique, resulting optimal powers are unique.

Next, we change the energy arrival profiles for the source and the relay as $\esource = [12 ;0 ;0 ;0]$ mJ and $\erelay = [5 ;1 ;0 ;2]$ mJ, respectively, with energy transfer efficiency $\alpha = 0.5$. Note that the source node is not energy harvesting. In this case, we find the optimal energy transfer vector as  $\deltav = [2.67; 0; 0; 0]$ mJ and the resulting optimal power vectors are $\prelay = \psource = [2.33 ;2.33 ;2.33 ;2.33]$ mW. Note that the optimal power sequences for the source and the relay match in this specific example, which does not hold in general.

\subsection{Numerical Example for the Gaussian Two-Way Channel}

In this section, we consider the Gaussian two-way channel model in Section~\ref{twcmodel}. The energy arrival profiles of user 1 and user 2 are $\esource = [5; 10; 5]$ mJ and $\erelay = [10; 5; 10]$ mJ, respectively, and the wireless energy transfer efficiency is set to $\alpha = 0.7$. We found the capacity region by running the two-dimensional directional water-filling algorithm for all $\hp_1,\hp_2 \geq 0$. We plot the resulting capacity region in Fig.~\ref{twccapsim}, where we also plot the capacity region when energy transfer is not allowed. Note that when energy transfer is not allowed, the capacity region is the rectangle with single-user optimal rates subject to the individual energy arrivals. We observe that the availability of wireless energy transfer significantly improves the capacity region.

\begin{figure}[t]
\begin{center}
\includegraphics[scale=0.75]{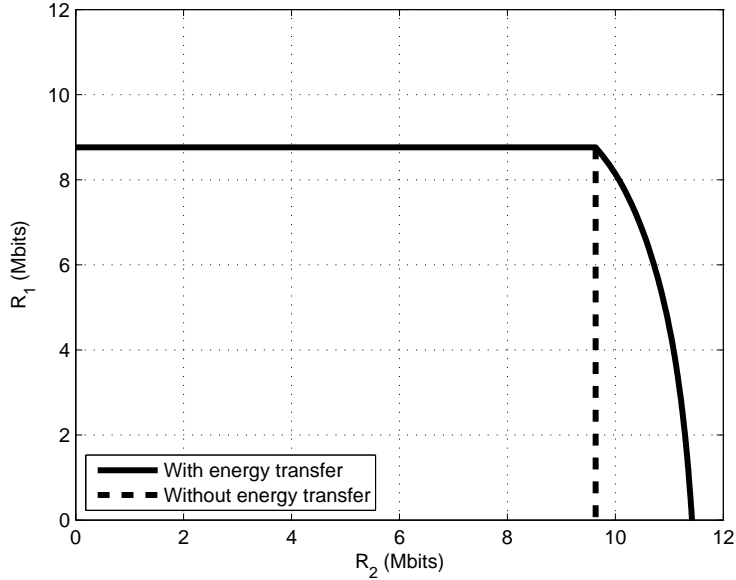}
\end{center}
\caption{Capacity region of the two-way channel with energy transfer.}
\label{twccapsim}
\vspace*{-0.1in}
\end{figure}

\subsection{Numerical Example for the Gaussian Multiple Access Channel}

In this section, we consider the Gaussian multiple access channel model in Section~\ref{macmodel}. The energy arrival profiles of user 1 and user 2 are $\esource = [5; 2; 5]$ mJ and $\erelay = [1; 3; 1]$ mJ, respectively, and wireless energy transfer efficiency is $\alpha = 0.5$. We plot the resulting capacity region in Fig.~\ref{maccapsim} and we compare it with the region when no energy transfer is allowed. Note that when no energy transfer is allowed, the region is found by the backward directional water-filling algorithm in \cite{jing12jcn}. We observe in Fig.~\ref{maccapsim} that the boundary of the capacity regions when energy transfer is allowed and not allowed match when the priority of user 1 is higher than the priority of user 2. However, the availability of wireless energy transfer significantly improves the capacity region when priority of user 2 is higher than the priority of user 1.

\begin{figure}[t]
\begin{center}
\includegraphics[scale=0.75]{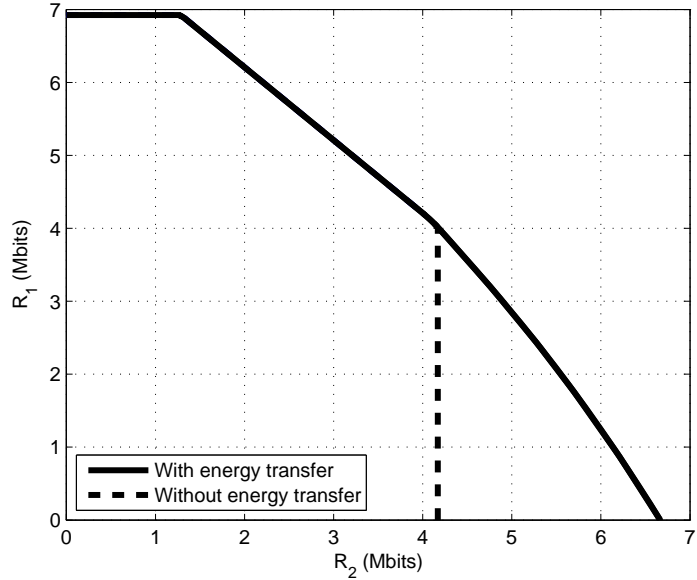}
\end{center}
\caption{Capacity region of the multiple access channel with energy transfer.}
\label{maccapsim}
\vspace*{-0.1in}
\end{figure}

\section{Concluding Remarks}

Energy cooperation made possible by wireless energy transfer is a fundamental shift in terms of the energy dynamics of a wireless network, yielding new performance limits. In this paper, we studied the communication performance of simple two- and three-node wireless networks in a deterministic setting where nodes harvest energy from the environment and wireless energy transfer is possible from one user to another in one-way and with efficiency $\alpha$. We first considered the Gaussian two-hop relay channel and studied the end-to-end throughput maximization problem. We showed that if the relay energy profile is higher first and then lower, the rates of the source and the relay nodes need to be matched in the optimal policy. We also showed that if the source is not energy harvesting, then transferring energy in the first slot is optimal. Next, we studied the capacity region of the Gaussian two-way channel. We showed that the boundary of the capacity region is achieved by policies that are given by a generalized version of two-dimensional directional water-filling algorithm. Finally, we studied the Gaussian multiple access channel. We showed that no energy transfer is needed if the priority of the first user is higher, and all of the energy needs to be transferred to the second user if the priority of the second user is sufficiently high. These results reveal new insights on how energy is optimally allocated in multi-user scenarios when wireless energy transfer is available as a new degree of freedom in network design. We remark that the analysis for finding the optimal policies in each multi-user setting can be extended for the cases when bi-directional energy transfer is allowed. In the two-hop relay setting, if bi-directional energy transfer is allowed, perfectly matching the energy profiles of the source and the relay nodes would be feasible and hence optimal: In this case, we collect energy arrivals of the source and the relay in a single energy queue and perform a single-user optimization. We then divide resulting power allocation equally for the source and the relay. Similarly, \cite{kaya13ita} recently presented the extension of the analysis for two-way and multiple access channels when bi-directional energy transfer is allowed.

\appendices

\section{Proof of Lemma~\ref{convexlemma}}
\label{pf1}

Consider two feasible power policies and energy transfer profiles $(\psource_1,\prelay_1,\deltav_1)$ and $(\psource_2,\prelay_2,\deltav_2)$. Let us consider a new policy as a convex combination of these two policies, i.e., $(\mathbf{P_{3}},\mathbf{\bar{P}_{3}}, \deltav_3)  = \lambda (\psource_1,\prelay_1,\deltav_1) + (1-\lambda) (\psource_2,\prelay_2,\deltav_2) $ for $0 < \lambda < 1$.  First we show that this new policy is feasible:
\begin{align}
\sum_{i=1}^k P_{3i} & =  \sum_{i=1}^k \lambda P_{1i} + (1-\lambda) P_{2i} \\
& \leq \lambda \sum_{i=1}^k  (E_{i} - \delta_{1i}) + (1-\lambda) \sum_{i=1}^k (E_{i} - \delta_{2i}) \\
& = \sum_{i=1}^k (E_{i} - \delta_{3i}), \quad k=1,\dots,T
\end{align}
We use similar arguments for $\bar{P}_{3i},\delta_{3i}$ and show that the policy $(\mathbf{P_{3}},\mathbf{\bar{P}_{3}}, \deltav_3)$ is feasible. 

Now, consider the upper corner points of the achievable rate regions for $(\psource_1,\prelay_1,\deltav_1)$ and $(\psource_2,\prelay_2,\deltav_2)$. Since $\log(1 + p)$ is concave in $p$, we have
\begin{align}
\sum_{i=1}^T \log(1 + P_{3i}) &>  \sum_{i=1}^T \lambda \log(1 + P_{1i}) + (1-\lambda) \sum_{i=1}^T \log(1+P_{2i}) \\
\sum_{i=1}^T \log(1 + \bar{P}_{3i}) &>  \sum_{i=1}^T \lambda \log(1 + \bar{P}_{1i}) + (1-\lambda) \sum_{i=1}^T \log(1+\bar{P}_{2i})
\end{align}
This means that the new policy $(\mathbf{P_{3}},\mathbf{\bar{P}_{3}}, \deltav_3)$ achieves a higher throughput for both users than the line connecting the two upper corner points under policies $(\psource_1,\prelay_1,\deltav_1)$ and $(\psource_2,\prelay_2,\deltav_2)$. Therefore, the region $\mycap$ is a convex region.

\section{Proof of Lemma~\ref{convexlemmamac}}
\label{pf2}

Consider two feasible power policies and energy transfer profiles $(\psource_1,\prelay_1,\deltav_1)$ and $(\psource_2,\prelay_2,\deltav_2)$. Let us consider a new policy as a convex combination of these two policies, i.e., $(\mathbf{P_{3}},\mathbf{\bar{P}_{3}}, \deltav_3)  = \lambda (\psource_1,\prelay_1,\deltav_1) + (1-\lambda) (\psource_2,\prelay_2,\deltav_2) $ for $0 < \lambda < 1$. Since the constraints in set $\mathcal{F}$ are linear in the power vectors, it can be shown as in the proof of Lemma~\ref{convexlemma} in Appendix~\ref{pf1} that this new policy is feasible. 

Now, let $S_i$ be the pentagon created by the policy $(\psource_i,\prelay_i,\deltav_i)$, for $i=1,2,3$. Choose $\boldsymbol{t}_1 \in S_1$ and $\boldsymbol{t}_2 \in S_2$ to form $\boldsymbol{t}_3 = \lambda \boldsymbol{t}_1 + (1-\lambda) \boldsymbol{t}_2$ for $0 \leq \lambda \leq 1$. We need to show that $\boldsymbol{t}_3 \in S_3$. We proceed as follows:
\begin{align}
t_{31}  &= \lambda t_{11} + (1-\lambda) t_{21}\\
&\leq \lambda \sum_{i=1}^T \log(1 + P_{1i}) + (1-\lambda) \sum_{i=1}^T \log(1+P_{2i}) \\
& \leq \sum_{i=1}^T \log(1 + \lambda P_{1i} + (1-\lambda) P_{2i}) \\
&= \sum_{i=1}^T \log(1 + P_{3i})
\end{align}
Similarly, we show $t_{32} \leq \sum_{i=1}^T \log(1 + \bar{P}_{3i})$. Finally 
\begin{align}
t_{31} + t_{32}  &= \lambda (t_{11} + t_{21}) + (1-\lambda) (t_{21} + t_{22}) \\
&\leq \lambda \sum_{i=1}^T \log(1 + P_{1i} + \bar{P}_{1i}) + (1-\lambda) \sum_{i=1}^T \log(1+P_{2i}+\bar{P}_{2i}) \\
& \leq \sum_{i=1}^T \log(1 + \lambda (P_{1i} + \bar{P}_{1i}) + (1-\lambda) (P_{2i} + \bar{P}_{2i})\\
& = \sum_{i=1}^T \log(1 + P_{3i} + \bar{P}_{3i})
\end{align}
These inequalities show that $\boldsymbol{t}_3 \in S_3$ since it satisfies the boundary conditions of $S_3$. Therefore, the region $\mycap$ is a convex region.

\end{document}